\documentclass[a4paper]{article}
\usepackage{amsfonts,amssymb,amsmath,amsthm,cite}
\usepackage{ucs} 
\usepackage[utf8x]{inputenc}
\usepackage{graphicx}
\usepackage[english]{babel}
\usepackage{slashed}
\usepackage{textcomp}
\usepackage{bm}

\evensidemargin=\oddsidemargin
\newtheorem{lemma}{Lemma}

\begin{document}
\begin{center}
	{\Large\textbf{Ordered Exponential and Its Features\\ in Yang--Mills Effective Action}}
	\vspace{0.5cm}
	
	{\large A.~V.~Ivanov$^\dag$ and N.~V.~Kharuk$^\ddag$}
	
	\vspace{0.5cm}
	
	{\it St. Petersburg Department of Steklov Mathematical Institute of
		Russian Academy of Sciences,}\\{\it 27 Fontanka, St. Petersburg 191023, Russia}\\
	{\it Leonhard Euler International Mathematical Institute, 10 Pesochnaya nab.,}\\
	{\it St. Petersburg 197022, Russia}\\
	$^\dag${\it E-mail: regul1@mail.ru}~~~~$^\ddag${\it E-mail: natakharuk@mail.ru}	
\end{center}
%\vskip 10mm
\date{\vskip 20mm}
%\begin{flushright}
%	\Large{{\calligra Dedicated to our parents}}
%\end{flushright}
%\vskip 10mm
\begin{abstract}
In this paper we discuss some non-trivial relations for ordered exponentials on smooth Riemannian manifolds. As an example of application, we study a dependence of the four-dimensional quantum Yang--Mills effective action on the background filed and gauge transformations. Also, we formulate some open questions about a structure of divergences.
\end{abstract}
\normalsize
\tableofcontents

\section{Introduction}
Ordered exponentials play a crucial role in mathematical and  theoretical physics, see \cite{Li,pol}, because they have explicit geometrical and physical meanings. Indeed, on the one hand side, they appear naturally in the differential geometry, see formulae (10.13a) and (10.17) in \cite{2}, as a solution 
for the "parallel transport" equation on principle bundles. On the other hand side, they allow us to make a transition to the Fock--Schwinger gauge condition in the non-Abelian gauge theories \cite{Shore}, and, hence, they are elements of the gauge transformation group. Also, formally, from the mathematical point of view, the ordered exponential is a product integral, which appears in different applications quite frequently \cite{Sl}. Of course, there are a lot of ways to apply the exponentials and an ordering itself, including the theory of integrable models \cite{co}, the theory of the heat kernel method \cite{31,33,333}, and many others \cite{bo}.

In this work, we present some non-trivial properties of the ordered exponential on a smooth Riemannian manifold in the case of a compact semisimple Lie group, see \cite{Shore,33,H}. We apply these properties to the pure four-dimensional Yang--Mills theory \cite{1,4,3} to make some useful conclusions about the properties of the effective action regarding its dependence on the background field and gauge transformations. We show that actually the action depends on the field stress tensor and its covariant derivatives. Also, we formulate several crucial questions about the structure of divergences, partially mentioned in \cite{21,22,Ivanov-2018}.

The paper has the following structure. In Section \ref{17-08-22:sec:ord}, we give the problem statement and formulate basic properties of ordered exponentials. Then, in Section \ref{17-08-22:sec:add}, we present some additional non-trivial properties and prove them. In Section \ref{17-08-22:sec:appl}, we discuss the application of the obtained equalities to the four-dimensional Yang--Mills effective action and formulate actual problems for further work. In Conclusion, we give some useful remarks and discuss the text of the work and alternative proofs.

\section{Ordered exponentials}
\label{17-08-22:sec:ord}
Let $G$ be a compact semisimple Lie group, and $\mathfrak{g}$ is its Lie algebra, see \cite{H}.
Let $t^a$ be the generators of the algebra $\mathfrak{g}$, where $a=1,\ldots,\dim\mathfrak{g}$,
such that the relations hold
\begin{equation}\label{constdef}
[t^a,t^b]=f^{abc}t^c,\,\,\,\,\,\,\mathrm{tr}(t^at^b)=-2\delta^{ab},
\end{equation}
where $f^{abc}$ are antisymmetric structure constants for $\mathfrak{g}$, and '$\mathrm{tr}$' is the Killing form. We work with an adjoint representation. It is easy to verify that the structure constants have the following crucial property
\begin{equation}\label{constprop}
f^{abc}f^{aef}=f^{abf}f^{aec}-f^{acf}f^{aeb}.
\end{equation}

Let $(M,g)$ be a smooth Riemannian manifold of dimension $\dim M=d\in\mathbb{N}$. Then, we introduce two elements $x,y\in U$, where $U\subset M$ is a smooth convex open domain, and Greek letters denote the coordinate components. Of course, we assume that the metric components form the symmetric matrix-valued operator, so we have $g^{\mu\nu}(x)=g^{\nu\mu}(x)$. The corresponding local formula for the Christoffel symbols of the second kind can be written as
\begin{equation}\label{deffs1}
\Gamma^{\rho}_{\mu\nu}(x)=\frac{g^{\rho\alpha}(x)}{2}
\bigg(\frac{\partial g_{\alpha\mu}(x)}{\partial x^{\nu}}+
\frac{\partial g_{\alpha\nu}(x)}{\partial x^{\mu}}-
\frac{\partial g_{\mu\nu}(x)}{\partial x^{\alpha}}\bigg).
\end{equation}

Further, by symbol $B^{\phantom{a}}_\mu(x)=B^a_\mu(x)t^a$, where $B^{\phantom{a}}_\mu(\cdot)\in C^{\infty}(U,\mathfrak{g})$ for all values of $\mu$, we define the components of a Yang--Mills connection 1-form, see \cite{2}. The operator $B^{\phantom{a}}_\mu(x)$ as an element of the Lie algebra acts by a commutator according to the adjoint representation. Hence, from now we consider $B^{\phantom{a}}_\mu(x)$ as the matrix-valued operator with the components $f^{adb}B^{d}_\mu(x)$.

For the purposes of the section we need to introduce several additional objects. First of all, we describe a number of derivatives. Let $h_{\nu_1\ldots\nu_n}^{\mu_1\ldots\mu_k}(\cdot)\in C^1(U,\mathfrak{g})$ be a tensor-valued operator, and $h_{\nu_1\ldots\nu_n}^{\mu_1\ldots\mu_k\,ab}(x)$ be its matrix components at the point $x$, then we define
\begin{equation}\label{deffs2}
\nabla_{x^\rho}h_{\nu_1\ldots\nu_n}^{\mu_1\ldots\mu_k}(x)=
\partial_{x^\rho}h_{\nu_1\ldots\nu_n}^{\mu_1\ldots\mu_k}(x)+
\sum_{i=1}^k\Gamma_{\rho\alpha}^{\mu_i}(x)h_{\nu_1\ldots\ldots..\nu_n}^{\mu_1\ldots\alpha\ldots\mu_k}(x)-
\sum_{i=1}^n\Gamma_{\rho\nu_i}^{\alpha}(x)h_{\nu_1\ldots\alpha\ldots\nu_n}^{\mu_1\ldots\ldots..\mu_k}(x),
\end{equation}
where $\alpha$ stands instead of the corresponding $i$-th index, and
\begin{align}\label{deffs}
	\overrightarrow{D}^{ab}_{x^{\rho}}h_{\nu_1\ldots\nu_n}^{\mu_1\ldots\mu_k\,bc}(x)&=\nabla_{x^{\rho}}^{\phantom{a}}h_{\nu_1\ldots\nu_n}^{\mu_1\ldots\mu_k\,ac}(x)+f^{adb}B^d_\rho(x)h_{\nu_1\ldots\nu_n}^{\mu_1\ldots\mu_k\,bc}(x),\\
	\label{deffs3}
	h_{\nu_1\ldots\nu_n}^{\mu_1\ldots\mu_k\,ab}(x)\overleftarrow{D}^{bc}_{x^{\rho}}&=\nabla_{x^{\rho}}^{\phantom{a}}h_{\nu_1\ldots\nu_n}^{\mu_1\ldots\mu_k\,ac}(x)-h_{\nu_1\ldots\nu_n}^{\mu_1\ldots\mu_k\,ab}(x)f^{bdc}B^d_\rho(x).
\end{align}

The main object of the paper, as it was noted in the introduction, is a path-ordered exponential. Let us give an appropriate definition by the following formula
\begin{multline}\label{expdef}
	\Phi^{ab}(x,y)=\delta^{ab}+\sum_{k=1}^{+\infty}(-1)^k
	\int_{u_0}^{u_1}ds_1\ldots\int_{u_0}^{s_{k-1}}ds_k\,
	\dot{\gamma}^{\mu_1}(s_1)
	\Big(f^{ad_1c_1}B^{d_1}_{\mu_1}(\gamma(s_1))\Big)\\\ldots\,
	\dot{\gamma}^{\mu_k}(s_k)
	\Big(f^{c_{k-1}d_kb}B^{d_k}_{\mu_k}(\gamma(s_k))\Big),
\end{multline}
where $\gamma(\cdot):[u_0,u_1]\to U$ is the parameterized geodesic, see \cite{pol}, such that $\gamma(u_0)=y$ and $\gamma(u_1)=x$, and the dot $\dot{\gamma}$ denotes the derivative $d\gamma(s)/ds$ by the parameter of parametrization. We note that the geodesic satisfies the following differential equation
\begin{equation}\label{deffs4}
\ddot{\gamma}^\rho(s)+\Gamma_{\mu\nu}^{\rho}(\gamma(s))\dot{\gamma}^\mu(s)
\dot{\gamma}^\nu(s)=0\,\,\,
\mbox{for all}\,\,\,s\in[u_0,u_1].
\end{equation}

Such type of operators have some useful properties, which can be formulated in the form
\begin{equation}\label{expprop}
	\Phi^{ab}(x,z)\Phi^{bc}(z,y)=\Phi^{ac}(x,y),\,\,
	(\Phi^{-1})^{ab}(x,y)=\Phi^{ab}(y,x)=\Phi^{ba}(x,y),\,\,
	\Phi^{ab}(y,y)=\delta^{ab},
\end{equation}
where the points $x,y,z\in U$ lie on the same geodesic. In other words, it means there is such a point $s\in\mathbb{R}$ that the equality $\gamma(s)=z$ holds. The proofs of the properties described above can be found in \cite{Shore,31,33}.

Then, we need to introduce Synge's world function $\sigma(x,y)$, see \cite{104}, as half the square of the geodesic length from $x$ to $y$. It is symmetric two-point scalar. Let us define some convenient abbreviations,
suggested in Synge's monograph \cite{1040}, as follows
\begin{align}
\label{deffs5}
\sigma_\rho(x,y)&=\partial_{x^\rho}\sigma(x,y),\,
\phantom{\partial_{y^{\rho^\prime}}\sigma^{\rho^\prime}}
\sigma^\rho(x,y)=g^{\rho\nu}(x)\sigma_\nu(x,y),\\
\label{deffs6}
\sigma_{\rho^\prime}(x,y)&=\partial_{y^{\rho^\prime}}\sigma(x,y),\,
\phantom{\partial_{x^\rho}\sigma^\rho}
\sigma^{\rho^\prime}(x,y)=\sigma_{\nu^\prime}(x,y)g^{\nu^\prime\rho^\prime}(y),
\end{align}
where the symbol "$\prime$\," means that we work with the use of the second argument of a two-point function. Let us note that if $U$ is a domain of $\mathbb{R}^d$ with $g^{\mu\nu}(x)=\delta^{\mu\nu}$, then Synge's world function $\sigma(x,y)$ is equal to $|x-y|^2/2$.

In the rest of the paper, in all derivatives (\ref{deffs2})--(\ref{deffs3}) we write $\rho$ and $\rho^\prime$ instead of $x^{\rho}$ and $y^{\rho^\prime}$, respectively, because it does not cause any confusion.
Therefore, we can formulate differential equations for exponential (\ref{expdef}) in the form
\begin{equation}\label{phieq}
	\sigma^\rho(x,y)\overrightarrow{D}^{ab}_{\rho^{\phantom{\prime}}}\Phi^{bc}(x,y)=0\,\,\,\mbox{and}\,\,\,\,
	\Phi^{ab}(x,y)\overleftarrow{D}^{bc}_{\rho^\prime}\sigma^{\rho^{\prime}}(x,y)=0.
\end{equation}

The proof can be achieved by the direct differentiation of (\ref{expdef}) and the integration by parts, see the papers mentioned above.

\section{Additional properties}
\label{17-08-22:sec:add}
In this section we present a number of formulae for representing the path-ordered exponential and some additional properties, which are useful in gauge theories. We start with two series representations, which actually generalize the covariant Taylor representation from \cite{30} on the case of non-zero Yang--Mills connection components.
\begin{lemma}\label{exprepr1}
Let $x,y\in U$, and $1$ denotes the unit function.
Then, under the conditions described above, we have
	\begin{align}\label{exprepr2}
		\Phi^{ab}(x,y)&=\delta^{ab}+\sum_{k=1}^{+\infty}\frac{(-1)^k}{k!}\sigma^{\mu_1^\prime}(x,y)\ldots\sigma^{\mu_k^\prime}(x,y)
		\Big(1\overleftarrow{D}^{ac_1}_{\mu_1^\prime}
		\overleftarrow{D}^{c_1c_2}_{\mu_2^\prime}\ldots
		\overleftarrow{D}^{c_{k-1}b}_{\mu_k^\prime}\Big),\\\label{exprepr3}
		\Phi^{ab}(x,y)&=\delta^{ab}+\sum_{k=1}^{+\infty}\frac{(-1)^k}{k!}
		\sigma^{\mu_k}(x,y)\ldots\sigma^{\mu_1}(x,y)
		\Big(\overrightarrow{D}^{ac_{k-1}}_{\mu_k^{\phantom{\prime}}}\ldots
		\overrightarrow{D}^{c_2c_1}_{\mu_2^{\phantom{\prime}}}
		\overrightarrow{D}^{c_{1}b}_{\mu_1^{\phantom{\prime}}}1\Big).
	\end{align}
\end{lemma}
\begin{proof}
	It is obvious that the initial condition $\Phi^{ab}(y,y)=\delta^{ab}$ holds, because $\sigma^{\rho}(x,y)$ and $\sigma^{\rho^\prime}(x,y)$ tend to zero when $x\to y$. Let us check that the right hand side of (\ref{exprepr2}) satisfies both equations from (\ref{phieq}). We start with the second one. 
	Let us apply the operator $\overleftarrow{D}^{bc}_{\rho^\prime}\sigma^{\rho^\prime}(x,y)$. 
	Due to the fact that the derivative acts at the point $y$, we get the answer instantly after using the following additional property of Synge's world function
\begin{equation}\label{deffs7}
\sigma^{\rho^\prime}(x,y)\partial_{y^{\rho^\prime}}\sigma^{\mu_i^\prime}(x,y)
=\sigma^{\mu_i^\prime}(x,y)+
\sigma^{\rho^\prime}(x,y)\Gamma_{\rho^\prime\nu^\prime}^{\mu_i^\prime}(y)\sigma^{\nu^\prime}(x,y),
\end{equation}
which actually follows from differentiation of $2\sigma(x,y)=\sigma^{\rho^\prime}(x,y)\sigma_{\rho^\prime}(x,y)$, see formula (2.31) in \cite{1040}.
	
Let us move on to the first equation from (\ref{phieq}). In the case we have the derivative at the point $x$. So we are going to use the covariant Taylor expansion in the form
	\begin{equation}\label{expanB}
		\sigma^{\rho}(x,y)B_\rho^d(x)=\sum_{k=1}^{+\infty}\frac{(-1)^k}{(k-1)!}
		\sigma^{\mu_k^\prime}(x,y)\ldots\sigma^{\mu_1^\prime}(x,y)
		\nabla_{\mu_k^\prime}^{\phantom{a}}\ldots\nabla_{\mu_2^\prime}^{\phantom{a}}B^d_{\mu_1^\prime}(y).
	\end{equation}
It does not follow from ordinary covariant expansion, because the last formula contains both, $\sigma^{\rho}(x,y)$ and $\sigma^{\rho^\prime}(x,y)$, and we need to decompose the first function in terms of another one. The simplest way to do this is to use the formalism of geodesic lines, successfully applied to prove the covariant expansions, see appendix in \cite{30}. First of all, let us rewrite the derivatives of Synge's world function as
\begin{equation}\label{deffs8}
\sigma^{\rho}(x,y)=(u_1-u_0)\dot{\gamma}^\rho(u_1)
,\,\,\,
\sigma^{\rho^\prime}(x,y)=-
(u_1-u_0)\dot{\gamma}^{\rho^\prime}(u_0).
\end{equation}
Hence, using the standard Taylor expansion, we get the chain of equalities
\begin{align}\label{deffs9}
\sigma^{\rho}(x,y)B_\rho^d(x)&=(u_1-u_0)\dot{\gamma}^\rho(u_1)B_\rho^d(\gamma(u_1))\\\label{deffs10}&=
\sum_{k=1}^{+\infty}\frac{(u_1-u_0)^k}{(k-1)!}\bigg(\frac{d^{k-1}}{dt^{k-1}}
\dot{\gamma}^\rho(t)B_\rho^d(\gamma(t))\bigg)\bigg|_{t=u_0}.
\end{align}
The expression in the large parentheses can be transformed in the manner
(an analog of (4.112) from \cite{30})
\begin{equation}\label{deffs11}
\bigg(\frac{d^{k-1}}{dt^{k-1}}
\dot{\gamma}^\rho(t)B_\rho^d(\gamma(t))\bigg)\bigg|_{t=u_0}=
\dot{\gamma}^{\mu_k^\prime}(u_0)\ldots\dot{\gamma}^{\mu_1^\prime}(u_0)
\nabla_{\mu_k^\prime}^{\phantom{a}}\ldots\nabla_{\mu_2^\prime}^{\phantom{a}}B^d_{\mu_1^\prime}(y),
\end{equation}
which can be proved by mathematical induction with the usage of (\ref{deffs4}).
Thereby, we obtain statement (\ref{expanB}) after applying the second relation from (\ref{deffs8}).

Further, after performing the differentiation of the right hand side of (\ref{exprepr2}), using the last equality and one additional property of Synge's world function
\begin{equation}\label{deffs12}
\sigma^{\rho}(x,y)\partial_{y^{\rho}}\sigma^{\mu_i^\prime}(x,y)
=\sigma^{\mu_i^\prime}(x,y),
\end{equation}
we get a number of relations for each degree of $\sigma^{\mu^\prime}(x,y)$
	\begin{multline}\label{recrel1}
		\sum_{k=1}^{+\infty}\frac{(-1)^k\sigma^{\mu_1^\prime}(x,y)\ldots\sigma^{\mu_k^\prime}(x,y)}{(k-1)!}
		\Bigg(1\overleftarrow{D}^{ac_1}_{\mu_1^\prime}
		\overleftarrow{D}^{c_1c_2}_{\mu_2^\prime}\ldots
		\overleftarrow{D}^{c_{k-1}c}_{\mu_k^\prime}+
		\nabla_{\mu_{k}^\prime}^{\phantom{a}}\ldots\nabla_{\mu_2^\prime}^{\phantom{a}}f^{adc}B^d_{\mu_1^\prime}(y)\\+
		\sum_{n=1}^{k-1}
		C_{n}^{k-1}
		\Big(\nabla^{\phantom{a}}_{\mu_{k}^\prime}\ldots\nabla^{\phantom{a}}_{\mu_{n+2}^\prime}f^{ade}B^d_{\mu_{n+1}^\prime}(y)\Big)
		\Big(1\overleftarrow{D}^{ec_1}_{\mu_1^\prime}
		\overleftarrow{D}^{c_1c_2}_{\mu_2^\prime}\ldots
		\overleftarrow{D}^{c_{k-1}c}_{\mu_n^\prime}
		\Big)\Bigg).
	\end{multline}
The last sum is equal to zero. Indeed, the factor $\sigma^{\mu_1^\prime}(x,y)\ldots\sigma^{\mu_k^\prime}(x,y)$ leads to the symmetrization of the tensor in large parentheses. This means, in particular, that we can change the order of derivatives. Hence, the first term cancels the others after factorizing of the first degree of the Yang--Mills connection.

Therefore, we have obtained the first statement of the lemma. The second equality follows from the first one, the second property from formula (\ref{expprop}), the following permutations $a,x\leftrightarrow b,y$, and the relation $f^{abc}=-f^{cba}$.
\end{proof}

Notice one more representation for the path-ordered exponential.
\begin{lemma}\label{lem2lem}
Let $x,y\in U$, and $1$ denotes the unit function. Also, $N_r(x,y)$ and $N_l(x,y)$ denote matrix-valued operators $\sigma^\rho(x,y)\overrightarrow{D}^{ab}_{\rho^{\phantom{\prime}}}$ and
$\overleftarrow{D}^{ab}_{\rho^\prime}\sigma^{\rho^{\prime}}(x,y)$, respectively.
Then, under the conditions described above, we have
\begin{equation}\label{newrep1}
\Phi(x,y)=\lim_{t\to1-0}
\bigg(e^{N_r(x,y)\ln(1-t)}\,1\bigg)=
\lim_{t\to1-0}
\bigg(1\,e^{N_l(x,y)\ln(1-t)}\bigg).
\end{equation}
\end{lemma}	
\begin{proof}
For simplicity we work with matrix-valued operators. Then, using property (\ref{deffs7}), we can write the following relation
\begin{equation}\label{newrep2}
\sigma^{\mu_k}(x,y)...\sigma^{\mu_1}(x,y)\Big(\overrightarrow{D}_{x^{\mu_1}}...\overrightarrow{D}_{x^{\mu_k}}\Big)=
\prod_{i=0}^{k-1}
\big(N_r(x,y)-i\big).
\end{equation}
It is obtained with the use of mathematical induction. Applying the change $x\leftrightarrow y$ and transposition of the matrices, we can get the relation for the left operators. Let us note, that the right hand side of (\ref{newrep2}) contains the product of commutative (with each other) operators.
	
Hence, we can rewrite the following chain of equalities for representation (\ref{exprepr3})
\begin{align}\label{newrep3}
\Phi(x,y)&=\sum_{k=0}^{+\infty}\frac{1}{k!}
\Bigg(\prod_{i=0}^{k-1}\big(i-N_r(x,y)\big)\Bigg)1
=\sum_{k=0}^{+\infty}\frac{\Gamma\big(k-N_r(x,y)\big)}{\Gamma\big(k+1\big)\Gamma\big(-N_r(x,y)\big)}\,1
\\\label{newrep4}&=\lim_{t\to1-0}
\Bigg(\sum_{k=0}^{+\infty}\frac{t^k\,\Gamma\big(k-N_r(x,y)\big)}{\Gamma\big(k+1\big)\Gamma\big(-N_r(x,y)\big)}\,1\Bigg)=
\lim_{t\to1-0}
\bigg((1-t)^{N_r(x,y)}\,1\bigg),
\end{align}
from which we obtain the first relation of the lemma. The second one follows from the change $x,r\leftrightarrow y,l$ and the matrix transposition.
\end{proof}	

From statements (\ref{exprepr2}) and (\ref{exprepr3}) it follows that for smooth connection components we have a covariant Taylor expansion for the path-ordered exponential (\ref{expdef}). Moreover, it allows us to write out the answer for the ordered exponential in terms of the field $B^d_\mu$ and its covariant derivatives. At the same time we need to emphasize that the first ordinary derivatives (left and right) of Synge's world function are uniquely related to each other with the usage of covariant Taylor series, see formula (\ref{expanB}) for $B^d_\rho(x)=\delta_{\rho i}$ with fixed $i\in\{1,\ldots,d\}$. Hence, the degrees of $\sigma^{\rho}(x,y)$ and degrees of $\sigma^{\rho^\prime}(x,y)$ lead to equivalent expansions for two-point functions, like the degrees of $(x-y)^\rho$ in $\mathbb{R}^d$.

Now we are ready to formulate the next relation.
\begin{lemma}\label{mainlem}
Let $x,y\in U$. Then, under the conditions described above, we have
	\begin{equation}\label{expprop2}
		f^{abc}\Phi^{ae}(x,y)\Phi^{bd}(x,y)\Phi^{cg}(x,y)=f^{edg},
	\end{equation}
	\begin{equation}\label{expprop3}
		\Phi^{a_1d}(x,y)\Phi^{a_2c}(x,y)f^{a_1ba_2}f^{a_3ba_4}\Phi^{a_3e}(x,y)\Phi^{a_4g}(x,y)
		=f^{dbc}f^{ebg}.
	\end{equation}
\end{lemma}
\begin{proof} Let us apply the operator $\sigma^\rho(x,y)\partial_\rho$ to the left hand side of  formula (\ref{expprop2}). In the case, the derivative acts on every $\Phi$-factor. So we get three terms.
	
Then we need to combine formulae (\ref{deffs}) and (\ref{phieq}) and substitute 
$-\sigma^\rho(x,y) f^{adb}B^d_\rho(x)\Phi^{bc}(x,y)$ instead of $\sigma^\rho(x,y)\partial_{x^\rho}\Phi^{ac}(x,y)$. After that we obtain three new terms
	\begin{align*}
		&-f^{abc}\big(\sigma^\rho(x,y) f^{aij}B^i_\rho(x)\Phi^{je}(x,y)\big)
		\Phi^{bd}(x,y)\Phi^{cg}(x,y)\\&
		-f^{abc}\Phi^{ae}(x,y)
		\big(\sigma^\rho(x,y) f^{bij}B^i_\rho(x)\Phi^{jd}(x,y)\big)\Phi^{cg}(x,y)\\&
		-f^{abc}\Phi^{ae}(x,y)
		\Phi^{bd}(x,y)\big(\sigma^\rho(x,y) f^{cij}B^i_\rho(x)\Phi^{jg}(x,y)\big).
	\end{align*}
	
Therefore, if we apply property (\ref{constprop}) to the first term, it can be converted into the second and the third terms, but with opposite signs. This means we have proven that
	\begin{equation}\label{derofconst}
		\sigma^\rho(x,y)\partial_{x^\rho}\Big(f^{abc}\Phi^{ae}(x,y)
		\Phi^{bd}(x,y)\Phi^{cg}(x,y)\Big)=0.
	\end{equation}
	
Now we need to use the representation for the path-ordered exponential derived above, see (\ref{exprepr3}). From this it follows that we can expand the ordered exponential in the covariant Taylor series in powers of $\sigma^\rho(x,y)$, because the field $B_\mu(x)$ is smooth by definition.  
	
Hence, formula (\ref{derofconst}) means that $f^{abc}\Phi^{ae}(x,y)
\Phi^{bd}(x,y)\Phi^{cg}(x,y)$ does not depend on the variable $x$. So we can take any convenient value of $x\in U$. If we choose $x=y$ and use the property from (\ref{expprop}), then we get the first statement of the lemma. The second statement is the consequence of the first one.
\end{proof}

The last lemma is devoted to a product of three path-ordered exponentials. For convenience, we need to define the field strength tensor, components of which are equal to
\begin{equation*}
F_{\mu\nu}^a(x)=\partial_{x^\mu}^{}B_\nu^a(x)-\partial_{x^\nu}^{}B_\mu^a(x)+f^{abc}B_\mu^b(x)B_\nu^c(x),
\end{equation*}
where $x\in U$.

\begin{lemma}\label{philem}
	Let $x,y,z\in U$. Then, under the conditions described above, we have the following relation
	\begin{equation}\label{phinew1}
		\Phi(z,x)\Phi(x,y)\Phi(y,z)=\Phi(x,y)\Big|_{B_\mu(\cdot)\to f_\mu(\cdot-z,z)},
	\end{equation}
	where 
\begin{equation}\label{SD44}
f^{ab}_\sigma(x-z,z)=\Phi^{ac}(z,x)
\big(\delta^{cd}\partial_{x^\sigma}+f^{ced}B^e_{\sigma}(x)\big)
\Phi^{db}(x,z).
\end{equation}	
Moreover, let $g^{\mu\nu}(\cdot)=\delta^{\mu\nu}$ for all points from $U$. Also, $\mathbf{1}$, $B_\mu$, and $F_{\mu\nu}$ denote the matrix-valued operators with the elements $\delta^{ab}$, $f^{adb}B^d_\mu$, and $f^{adb}F^d_{\mu\nu}$, respectively.
Also, we define an additional derivative of a special type by the formula $\mathfrak{D}_{x^\mu}\cdot=\partial_{x^\mu}\cdot+[B_\mu(x),\cdot\,]$. So, we obtain
	\begin{equation}\label{SD4}
		f^{ab}_\sigma(x-z,z)=\sum_{k=1}^{+\infty}
		\frac{(x-z)^{\mu_1\ldots\mu_k}}{(k-1)!(k+1)}\mathfrak{D}^{ac_1}_{z^{\mu_1}}\ldots
		\mathfrak{D}^{c_{k-2}c_{k-1}}_{z^{\mu_{k-1}}}f^{c_{k-1}c_kb}F^{c_k}_{\mu_k\sigma}(z),
	\end{equation}
where $(x-z)^{\mu_1\ldots\mu_k}=(x-z)^{\mu_1}\cdot\ldots\cdot(x-z)^{\mu_k}$.
The first terms in decomposition (\ref{phinew1}) have the form
	\begin{align}\nonumber
		\Phi(z,x)\Phi(x,y)\Phi(y,z)=\mathbf{1}&+\frac{1}{2}(x-z)^\nu(y-z)^\mu F_{\nu\mu}(z)+
		\frac{1}{6}(z-y)^\nu(z-x)^\sigma(y-x)^\mu \mathfrak{D}_{z^{(\nu}}F_{\sigma)\mu}(z)\\\nonumber&+
		\frac{1}{8}(x-z)^{\nu\sigma}(y-z)^{\mu\rho} F_{\nu\mu}(z)F_{\sigma\rho}(z)-
		\frac{1}{24}(y-x)^{\mu\nu\sigma}(x-z)^\rho\mathfrak{D}_{z^{\mu}}\mathfrak{D}_{z^{\nu}}F_{\sigma\rho}(z)
		\\\label{expprop1}&-\frac{1}{16}(x-z)^{\mu\rho}(y-x)^{\sigma}
		(y-z)^\nu\mathfrak{D}_{z^{(\mu}}\mathfrak{D}_{z^{\nu)}}F_{\sigma\rho}(z)+\ldots,
	\end{align}
where "three dots" denotes terms, which have the total degree of a monomial more than four. Parentheses denote index symmetrization without division by the corresponding factorial.
\end{lemma}
\begin{proof} Let us introduce two gauge transformed derivatives according to the formulae
\begin{align}
		\label{SD3}
		\overrightarrow{D}_{x^\sigma}^{ab}(z)&=\Phi^{ac}(z,x)\overrightarrow{D}_{x^\sigma}^{cd}\Phi^{db}(x,z)=\delta^{ab}\overrightarrow{\partial}_{x^\sigma}+f^{ab}_\sigma(x-z,z),
		\\\label{SD7}
		\overleftarrow{D}_{x^\sigma}^{ab}(z)&=\Phi^{ac}(z,x)\overleftarrow{D}_{x^\sigma}^{cd}\Phi^{db}(x,z)=\delta^{ab}\overleftarrow{\partial}_{x^\sigma}-f^{ab}_\sigma(x-z,z),
	\end{align}	
and rewrite the equations from (\ref{phieq}) in the following form
	\begin{equation*}
		\sigma^\mu(x,y)\overrightarrow{D}^{ab}_{\mu}(z)\big(\Phi^{be}(z,x)\Phi^{ed}(x,y)\Phi^{dc}(y,z)\big)=0,\,\,\,
		\big(\Phi^{ab}(z,x)\Phi^{be}(x,y)\Phi^{ed}(y,z)\big)\overleftarrow{D}^{dc}_{\mu^\prime}(z)\sigma^{\mu^\prime}(x,y)=0.
	\end{equation*}	
Moreover, the combination $\Psi^{ac}(x,y;z)=\Phi^{ab}(z,x)\Phi^{bd}(x,y)\Phi^{dc}(y,z)$ satisfies the following initial condition $\Psi^{ac}(y,y;z)=\delta^{ac}$, which is the consequence of (\ref{expprop}). This means that $\Psi^{ac}(x,y;z)$ is the ordered exponential for the connection components of the form $f^{ab}_\sigma(\cdot-z,z)$. Hence, we have obtained statement (\ref{phinew1}). Formula (\ref{SD4}) is the particular case and has been derived in the paper \cite{33}. Then, expansion (\ref{expprop1}) follows from result (\ref{phinew1}) and the definition from (\ref{expdef}) with the use of the explicit form for geodesic $\gamma(s)=y+s(x-y)$.
\end{proof}

\section{Application to the Yang--Mills theory}
\label{17-08-22:sec:appl}
Let us introduce an effective action for the four-dimensional quantum Yang--Mills theory \cite{1,3}. For this we use the main conditions of Section \ref{17-08-22:sec:ord}, but with some additional restrictions: $M=\mathbb{R}^4$, $g^{\mu\nu}(\cdot)=\delta^{\mu\nu}$ for all points from $\mathbb{R}^4$, and the field $B_\mu(x)$ satisfies the quantum equation of motion, which follows from the use of the path integral formulation \cite{3,23} and the background field method \cite{102,103,24,25,26}. This field is called background field.

For further work we need to define some additional constructions, such as the classical action of the Yang--Mills theory
\begin{equation}\label{Clact}
	W_{-1}[B]=\frac{1}{4}\int_{\mathbb{R}^4}d^4x\,
	F^a_{\mu\nu}(x)F^{a}_{\mu\nu}(x),
\end{equation}
the following two Laplace-type operators
\begin{equation}\label{oper}
	M_0^{ab}(x)=-\overrightarrow{D}_{x_\mu}^{ae}\overrightarrow{D}_{x^\mu}^{eb},\,\,\,
	M_{1\mu\nu}^{\,\,\,ab}(x)=M_0^{ab}(x)\delta_{\mu\nu}^{}-2f^{acb}F_{\mu\nu}^c(x),
\end{equation}
and the corresponding Green's functions $G_0$ and $G_1$, which follow from the equalities
\begin{equation}
\label{green}
M_{1\mu\nu}^{\,\,\,ab}(x)G_{1\nu\rho}^{\,\,\,bc}(x,y)=\delta^{ac}\delta_{\mu\rho}\delta(x-y),\,\,\,\,\,\,
M_0^{ab}(x)G_{0}^{bc}(x,y)=\delta^{ac}\delta(x-y),
\end{equation}
with appropriate boundary conditions. These conditions have a physical nature and, actually, are not studied well enough. They should be such that the problem for the quantum equation of motion would be well posed.

Now we are ready to introduce the effective action of the Yang--Mills theory for small values of a coupling constant $g$. This is a function of the background field $B_\mu$, which has the following asymptotic expansion for $g\to+0$
\begin{equation}
\label{eq20}
W[B,\Lambda]=\frac{1}{g^2}W_{-1}[B]+\bigg\{\frac{1}{2}\ln\det\big(M_1^{\Lambda}/M_1^{\Lambda}|_{B=0}\big)-\ln\det\big(M_0^{\Lambda}/M_0^{\Lambda}|_{B=0}\big)\bigg\}+
\sum_{n=1}^{+\infty}g^{2n}W_n[B,\Lambda],
\end{equation}
where the symbol $\Lambda$ denotes some type of regularization, such that the Green's functions are deformed near the diagonal ($x\sim y$). Removing the regularization corresponds to the limit $\Lambda\to+\infty$. Using the language of quantum field theory, we can declare that $W_n[B,\Lambda]$ is the $(n+1)$-th quantum correction, corresponding to the $(n+1)$-loop contribution, see \cite{21,22}. Precise formulae for these corrections can be found in papers \cite{23,Ivanov-Kharuk-2020}. Fortunately, explicit expressions do not matter in our calculations, that is why we have introduced the effective action in such a general form.

Let us draw attention that the regularization is necessary, because all corrections to the classical action in (\ref{eq20}) contain divergent integrals. Some first terms have been studied earlier. For example, the divergent part of  "$\ln\det$" has been calculated explicitly in \cite{16,17,12}. The two-loop contribution also has been computed with the use of different regularizations, see \cite{Ivanov-Kharuk-2020,12,13, Iv-Kh-2022}.

Then, we should note that the corrections from the last sum can be constructed with the use of the integration operator over $\mathbb{R}^4$ and the following elementary blocks (with the regularization applied)
\begin{equation}\label{gendiag5000}
\begin{tabular}{cccc}
		$G_{0}^{ab}(x_i,x_j),$&
		$\overrightarrow{D}^{ac}_{x_i^\mu}G_{0}^{cb}(x_i,x_j),$&
		$G_{0}^{ac}(x_i,x_j)\overleftarrow{D}^{cb}_{x_j^\sigma},$&
		$\overrightarrow{D}^{ae}_{x_i^\mu}G_{0}^{ec}(x_i,x_j)\overleftarrow{D}^{cb}_{x_j^\sigma},$\\
		$G_{1\nu\rho}^{\,\,\,ab}(x_i,x_j),$&
		$\overrightarrow{D}^{ac}_{x_i^\mu}G_{1\nu\rho}^{\,\,\,cb}(x_i,x_j),$&
		$G_{1\nu\rho}^{\,\,\,ac}(x_i,x_j)\overleftarrow{D}^{cb}_{x_j^\sigma}$,&
		$\overrightarrow{D}^{ae}_{x_i^\mu}G_{1\nu\rho}^{\,\,\,ec}(x_i,x_j)\overleftarrow{D}^{cb}_{x_j^\sigma},$
\end{tabular}
\end{equation}
which are connected to each other only by the following combinations of the structure constants:
\begin{equation}
f^{abc}\,\,\,\mbox{and}\,\,\,f^{abe}f^{ecd}.
\end{equation}
It is quite important to note two conditions. Firstly, all the variables should be under the integration. Secondly, if three (four) basic blocks are connected to each other with the usage of $f^{abc}$ ($f^{abe}f^{ecd}$), then the corresponding ends of the blocks have the same variables. Other combinations do not exist.
Also, we have noted above that the blocks from (\ref{gendiag5000}) should contain regularized Green's functions instead of $G_0$ and $G_1$. Now we want to introduce one restriction on a type of regularization. We assume that regularized Green's functions respect covariance. This means that after applying the gauge transformation, we obtain
\begin{equation}\label{gt1}
B_{\mu}^{\phantom{\omega}}(x)\to B_{\mu}^{\omega}(x)=
\omega^{-1}(x)B_{\mu}^{\phantom{\omega}}(x)\omega(x)+
\omega^{-1}(x)\partial_{x^\mu}^{\phantom{\omega}}\omega(x)
\end{equation}
and the regularized Green's functions get the following changes
\begin{align}\label{gt2}
G_{0}^{ab}(x_i,x_j)\Big\vert^{\scriptsize{\Lambda\mbox{-reg.}}}&\to
\bigg(G_{0}^{ab}(x_i,x_j)\Big\vert^{\scriptsize{\Lambda\mbox{-reg.}}}_{B_{\mu}^{\phantom{\omega}}\to B_{\mu}^{\omega}}\bigg)=
\omega^{-1}(x_i)\bigg(G_{0}^{ab}(x_i,x_j)\Big\vert^{\scriptsize{\Lambda\mbox{-reg.}}}\bigg)\omega(x_j),\\
\label{gt3}
G_{1\nu\rho}^{\,\,\,ab}(x_i,x_j)\Big\vert^{\scriptsize{\Lambda\mbox{-reg.}}}&\to
\bigg(G_{1\nu\rho}^{\,\,\,ab}(x_i,x_j)\Big|^{\scriptsize{\Lambda\mbox{-reg.}}}_{B_{\mu}^{\phantom{\omega}}\to B_{\mu}^{\omega}}\bigg)=
\omega^{-1}(x_i)\bigg(G_{1\nu\rho}^{\,\,\,ab}(x_i,x_j)\Big\vert^{\scriptsize{\Lambda\mbox{-reg.}}}\bigg)\omega(x_j),
\end{align}
where $\omega(\cdot)\in C^{\infty}(\mathbb{R}^4,G)$. Of course, non-regularized functions have this property by construction. Using the fact that in further calculations we want to make a special type of the gauge transformation, we conclude that for simplicity we can produce all manipulations with (\ref{gendiag5000}) instead of the regularized ones, because the deformation does not break the covariance. Also, we need to draw attention that such regularizations exist. As an example, we can suggest the explicit cutoff regularization in coordinate representation, recently studied in \cite{Ivanov-2022}.

For further study, we need to give some basic concepts and results on the heat kernel expansion, \cite{28,32,30,vas1}. Let us introduce a Laplace-type operator $A$, which has a more general view than one from (\ref{oper}). Locally, it has the following form
\begin{equation}\label{oper1}
A^{ab}(x)=IM_0^{ab}(x)-v^{ab}(x),
\end{equation}
where $I$ is an $n\times n$ unit matrix with $n\in\mathbb{N}$, and $v^{ab}(x)$ is an $n\times n$ matrix-valued smooth potential, such that the operator $A$ is symmetric. If we take $n=4$, $(I)_{\mu\nu}=\delta_{\mu\nu}$, and $(v^{ab})_{\mu\nu}(x)=2f^{acb}F^d_{\mu\nu}(x)$, then we obtain the operator $M_{1\mu\nu}^{\,\,\,ab}(x)$. Also, for the convenience we will not write the matrix $I$ in the rest of the text, because this does not make any confusion.

Then we move on to a definition of the local heat kernel. As we know, the standard heat kernel is the solution of the problem
\begin{equation}\label{HK1}
\big(\delta^{ac}\partial_\tau+A^{ac}(x)\big)K^{cb}(x,y;\tau)=0,\,\,\,K^{ab}(x,y;0)=\delta^{ab}\delta(x-y),
\end{equation}
paired with boundary conditions for (\ref{oper1}). At the same time, the local heat kernel \cite{iv-kh-2022} is defined as a solution of (\ref{HK1}), which for small values of the proper time $\tau\to+0$ has the following asymptotic expansion
\begin{equation}\label{HK2}
K^{ab}(x,y;\tau)=(4\pi\tau)^{-2}e^{-|x-y|^2/4\tau}\sum_{k=0}^{+\infty}\tau^k\mathfrak{a}_k^{ab}(x,y).
\end{equation}
We note that an asymptotic expansion for the ordinary heat kernel of $A$ can be presented as the sum of (\ref{HK2}) and some additional corrections. The coefficient $\mathfrak{a}^{ab}(x,y)$ of expansion (\ref{HK2}) is Seeley--DeWitt coefficient (also, it is named after Hadamard, Minakshisundaram \cite{111}, and Gilkey \cite{32}) coefficients, see \cite{110,1000}, and can be calculated recurrently \cite{31,33,32}, because they satisfy the following system of equations
\begin{equation}\label{SD1}
\mathfrak{a}_{0}^{ab}(x,y)=\Phi^{ab}(x,y),\,\,\,
\big(k+(x-y)^\sigma\overrightarrow{D}_{x^\sigma}^{ac}\big)
\mathfrak{a}_{k}^{cb}(x,y)=-A^{ac}(x)
\mathfrak{a}_{k-1}^{cb}(x,y),\,\,\,k\geqslant 1.
\end{equation}

Then, we are going to make a substitution. Let $z\in\mathbb{R}^4$ be an auxiliary point, that will play a role of an additional parameter. Then, we define
\begin{equation}\label{SD2}
\mathfrak{a}_{k}^{ab}(x,y;z)=
\Phi^{ac}(z,x)\mathfrak{a}_{k}^{cd}(x,y)\Phi^{db}(y,z).
\end{equation}

Further, we substitute formula (\ref{SD2}) into relation (\ref{SD1}) and introduce gauge transformed operators, according to the papers \cite{33,333}, in the form
\begin{align}\label{SD5}
M_0^{ab}(x-z;z)&=\Phi^{ac}(z,x)M_0^{cd}(x)\Phi^{db}(x,z)=-
\overrightarrow{D}_{x_\sigma}^{ac}(z)\overrightarrow{D}_{x^\sigma}^{cb}(z),
\\\label{SD6}
A^{ab}_z(x)&=\Phi^{ac}(z,x)A^{cd}(x)\Phi^{db}(x,z)
\\\nonumber&=
\delta_{\mu\nu}M_0^{ab}(x-z;z)-
\sum_{k=0}^{+\infty}
\frac{(x-z)^{\mu_1\ldots\mu_k}}{k!}\mathfrak{D}^{ac_1}_{z^{\mu_1}}\ldots
\mathfrak{D}^{c_{k-1}c_{k}}_{z^{\mu_{k}}}v^{c_{k}b}(z),
\end{align}
where we have used the Fock--Schwinger connection components (\ref{SD4}) and formula (\ref{oper}). Hence, the system of the equations from (\ref{SD1}) can be rewritten in the following form
\begin{align}\label{SD8}
\mathfrak{a}_{0}^{ab}(x,y;z)&=\Phi^{ac}(z,x)\Phi^{cd}(x,y)\Phi^{db}(y,z),\\\nonumber
\big(k+(x-y)^\sigma\overrightarrow{D}_{x^\sigma}^{ac}(z)\big)
\mathfrak{a}_{k}^{cb}(x,y;z)&=-A^{ac}_z(x)
\mathfrak{a}_{k-1}^{cb}(x,y;z),\,\,\,k\geqslant 1.
\end{align}

Now we are ready to note a very useful information. It follows from formulae (\ref{SD4}), (\ref{SD3}), (\ref{SD7}), (\ref{SD5}), and (\ref{SD6}), that if the potential $v$ is a function only of the field strength and its $\mathfrak{D}$-derivatives, then the decomposition of the Seeley--DeWitt coefficients $\mathfrak{a}_{k}^{cb}(x,y;z)$ can be constructed by using only the field stress tensor and its $\mathfrak{D}$-derivatives. As we can see, the operators $M_0^{ab}$ and $M_{1\mu\nu}^{\,\,\,ab}$ for the Yang--Mills theory are suitable for this statement. Additionally, we emphasize that formulae (\ref{SD1})--(\ref{SD8}) are valid for all dimensions.

Further, using the formulae introduced above, we can write out the following asymptotic behaviour near the diagonal for the Green's function in the four-dimensional space \cite{1000,29}
\begin{multline}\label{SD9}
\big(A^{-1}_z\big)^{ab}(x,y)=R_0(x-y)\mathfrak{a}_{0}^{ab}(x,y;z)+
R_1(x-y)\mathfrak{a}_{1}^{ab}(x,y;z)\\+
R_2(x-y)\mathfrak{a}_{2}^{ab}(x,y;z)+\mathcal{N}^{ab}_z(x,y)+\mathcal{ZM}^{ab}_z(x,y),
\end{multline}
where
\begin{equation}\label{SD10}
R_0(x)=\frac{1}{4\pi^2|x|^2},\,\,\,
R_1(x)=-\frac{\ln(|x|^2\mu^2)}{16\pi^2},\,\,\,
R_2(x)=\frac{|x|^2\big(\ln(|x|^2\mu^2)-1\big)}{64\pi^2},
\end{equation}
$\mathcal{N}^{ab}_z$ is a non-local part \cite{29}, depending on the boundary conditions of a spectral problem, and $\mathcal{ZM}^{ab}_z$ is a number of local zero modes to satisfy the problem. Let us separately note, it was shown in the paper \cite{Kharuk-2021}, that an infrared part in the second loop does not depend on $\mathcal{ZM}^{ab}_z$. Moreover, in the calculation process, we can choose $\mathcal{ZM}^{ab}_z$ in such a way \cite{Ivanov-Kharuk-2020}, that the non-local part $\mathcal{N}^{ab}_z$ would have the following behaviour near the diagonal $x\sim y$
\begin{equation}\label{SD11}
\mathcal{N}^{ab}_z(x,y)=-\frac{|x-y|^2}{2^7\pi^2}\mathfrak{a}_{2}^{ab}(y,y)\big(1+o(1)\big).
\end{equation}

Now we move on to the property mentioned in the introduction. For this we should show the possibility of the following replacement $f^{adb}B^d_\mu(\cdot)\to f^{ab}_\mu(\,\cdot-z,z)$ in the effective action, where $z\in\mathbb{R}^4$ is an auxiliary point and $f^{ab}_\mu$ is from (\ref{SD4}). The procedure can be carried out with the application of the gauge transformation with the usage of the ordered exponential $\Phi^{ab}(\,\cdot\,,z)$. In the case we have 
formulae (\ref{SD3}), (\ref{SD7}), and
\begin{align}\label{gendiag3}
G_{0}^{ab}(x_i,x_j)&=
\Phi^{ac}(x_i,z)G_{0}^{cd}(x_i,x_j;z)\Phi^{db}(z,x_j),\\\label{gendiag4}
G_{1\nu\rho}^{\,\,\,ab}(x_i,x_j)&=
\Phi^{ac}(x_i,z)G_{1\nu\rho}^{\,\,\,cd}(x_i,x_j;z)\Phi^{db}(z,x_j).
\end{align}
Using the last substitutions and formulae (\ref{green}), (\ref{SD5}), and (\ref{SD6}), we can conclude, that the new covariant derivatives $\overrightarrow{D}^{ab}_{x_i^\mu}(z)$ and $\overleftarrow{D}^{ab}_{x_j^\mu}(z)$ and the Green's functions $G_{0}^{ab}(x_i,x_j;z)$ and $G_{1\nu\rho}^{\,\,\,ab}(x_i,x_j;z)$ are constructed only with the use of the following connection components $f^{ab}_\mu(\,\cdot-z,z)$, see (\ref{SD4}), which are functions of the field stress tensor $F_{\mu\nu}$ and its $\mathfrak{D}$-derivatives.

At the same time after the gauge transformation, every building block from (\ref{gendiag5000}) is multiplied by the ordered exponential $\Phi^{ea}(x_i,z)$ on the left hand side and by $\Phi^{bg}(z,x_j)$ on the right hand side. Further, let us remember, that the top indices of the blocks should be connected to each other by the $f^{abc}$ or $f^{abe}f^{ecd}$. Hence, using the results of Lemma \ref{mainlem}, we can exclude all additional factors $\Phi^{ea}(x_i,z)$ and $\Phi^{bg}(z,x_j)$ from the blocks.
Therefore, we obtain the following equalities
\begin{equation}\label{gendiag2}
W_n\big[B,\Lambda\big]=W_n\big[f(\,\cdot-z,z),\Lambda\big]
\,\,\,\mbox{for all}\,\,\,n\geqslant1,
\end{equation}
where the functionals actually do not depend on  $z$. In other words, the procedure described above can be understood as re-expansion of the blocks near some auxiliary point. In the same manner we replace the operators in "$\ln\det$" by the gauge transformed ones. Then, using the invariance of the classical action with respect to the gauge transformations, we get
\begin{equation}\label{gendiag5}
	W\big[B,\Lambda\big]=W\big[f(\,\cdot-z,z),\Lambda\big].
\end{equation}

In the last part of the paper we want to formulate some actual open questions and make several useful remarks about divergencies searching. It is known that infrared divergences in the coordinate representation (or ultraviolet in the momentum one) follow from the appearance of non-integrable densities in (\ref{eq20}), because the behaviour (\ref{SD9}) of the Green's function near the diagonal ($x\sim y$) includes such functions as (\ref{SD10}).

At the same time, decomposition (\ref{SD9}) is applicable only in the vicinity of the diagonal, so we need to split the Green's function into two parts, in some neighbourhood of the $x=y$ and out of it. Formally, we can write this as
\begin{equation}\label{singord3}
\big(A^{-1}\big)(x,y)=
\underbrace{\big(A^{-1}\big)(x,y)\chi_{(\mu_1|x-y|<1)}}_{\mbox{\footnotesize "in" part}}+
\underbrace{\big(A^{-1}\big)(x,y)\chi_{(\mu_1|x-y|\geqslant1)}}_{\mbox{\footnotesize "out" part}},
\end{equation}
where $\mu_1>0$ is an auxiliary fixed  finite number of momentum dimension, and $\chi_{(\mbox{\footnotesize inequality})}$ is the Heaviside step function of a domain from $\mathbb{R}^4$. While for the first part we can write the standard decomposition (\ref{SD9}) near the diagonal, the second part, actually, is unknown to us, because it and its properties depend on the boundary conditions and the form of the background field. Anyway, we assume that "out" part has a good enough behaviour at infinity, so the effective action does not include other type of divergences. Otherwise, we need to introduce one more regularization and study these non-infrared singularities in the coordinate representation.

Now we are ready to formulate some questions.
\begin{itemize}
	\item Does the effective action contain the infrared singularity in coordinate representation, depending on the "out" part? For the positive answer, it is sufficient to provide explicit expression for some $W_n(B,\Lambda)$. We draw attention that the two-loop contribution $W_1(B,\Lambda)$ includes only singularities, related to the "in" part, because in this case the quantum correction is split into a sum of "in" and "out" parts, without mixing. Hence, without any fundamental justifications and assumptions, we are interested in the three-loop calculation. Moreover, it would be meaningful to describe the structure of such contributions for a cutoff regularization.
	\item Does the effective action contain the infrared singularity in coordinate representation, depending on the non-local part $\mathcal{N}^{ab}$? This is one more non-obvious issue. The fact is that in the two-loop contribution a special combination of non-local terms occurs, but it is converted to the local part with the use of the Seeley--DeWitt coefficients. The appearance or absence of the non-local part in the multi-loop terms is quite an interesting challenge, because it leads to the possibility of using simplifications in the calculations.
\end{itemize}
The last two questions are related to the form of the coefficients, their locality and dependence on the boundary conditions. However, we have a problem, corresponding to a structure of singularities itself. Indeed, in the quantum corrections after the regularization introduced we have some dimensional parameters: $\Lambda$ is the parameter of regularization, $\mu$ is the parameter from (\ref{SD10}), controlling presence of local zero modes in the Green's functions, and $\mu_1$ is the auxiliary parameter from (\ref{singord3}), which actually does not lead to any new dependence. Hence, the answer contain some dimensionless combinations of them. The question is the following.
\begin{itemize}
	\item Let $i,j\in\mathbb{N}\cup\{0\}$ and $n>0$. What combinations of $(\Lambda/m)^i\ln^j(\Lambda/m)$ does the correction $W_n(B,\Lambda)$ contain, where $m\in\{\mu,\mu_1\}$ is a dimension parameter? For example, it was shown \cite{Ivanov-Kharuk-2020,Iv-Kh-2022} that the two-loop correction includes only $\ln(\Lambda/\mu)$. And it is expected that for other orders we have $i=0$ and $j\in\{1,\ldots,n\}$.
\end{itemize}
If the last hypothesis is correct, then for the local singular contribution, depending on the "in" part, we have only one type of the coefficients with non-zero trace
$f^{aec}F^e_{\sigma\rho}f^{cdb}F^d_{\mu\nu}$.

\section{Conclusion}
\label{17-08-22:sec:conc}
In this paper, we have studied some non-trivial properties of the path-ordered exponentials. Additionally, we have shown that such equalities give the ability to make some useful conclusions about the structure of the four-dimensional Yang--Mills effective action. We believe that the result can be useful both in the general theory of renormalization of the Yang--Mills theory and in the investigation of some other models with non-Abelian structures. For example, in the sigma-models. 

Also, we need to repeat one more time a very important observation. The covariance under the gauge transformations of Feynman diagram blocks actually does not exclude non-logarithmic divergencies. To formulate the proposition in a stronger form we need some additional investigations, which are not from the scope of this work. They are formulated in the previous section as open questions.

In addition, we need to comment one more way to prove Lemma \ref{mainlem}. Let us introduce three auxiliary arbitrary elements $h_i=h_i^at^a$, where $i=1,2,3$, of the Lie algebra $\mathfrak{g}$. Then, an element $g$ of the Lie group $G$ acts on this algebra according to the adjoint representation as $\mathrm{Ad}_g(h_i)=gh_ig^{-1}$. Let $g^{ab}$ denotes the corresponding matrix elements, such that $\big(\mathrm{Ad}_g(h_i)\big)^at^a=g^{ab}h_i^bt^a$. After that we can write the following chain of equalities
\begin{align*}
f^{abc}\big(g^{ai}h_1^i\big)\big(g^{bj}h_2^j\big)\big(g^{ck}h_3^k\big)&=
-\frac{1}{2}\mathrm{tr}\big([\mathrm{Ad}_g(h_1),\mathrm{Ad}_g(h_2)]\mathrm{Ad}_g(h_3)\big)\\&=
-\frac{1}{2}\mathrm{tr}\big(\mathrm{Ad}_g([h_1,h_2])\mathrm{Ad}_g(h_3)\big)\\&=
-\frac{1}{2}\mathrm{tr}\big([h_1,h_2]h_3\big)=
f^{abc}h_1^ah_2^bh_3^c,
\end{align*}
from which we obtain an alternative proof. Despite the fact that this derivation looks shorter, we note that it includes the usage of some additional properties, presented in (\ref{expprop}). Indeed, we need to make sure that there is such an element $g$ in $G$ that the relation $g^{ab}=\Phi^{ab}(x,y)$ holds. So, we believe that the proof of Lemma \ref{mainlem}, written above, is more elegant and instructive in the context of covariant Taylor expansions.

Let us give some comments about the Yang--Mills effective action. In the paper, we have used the definition without using the path ordered formulation, because the last construction contains a lot of mathematical questions not related to the main topic. Due to this, we did not face the discussion of the measure properties and its existence. Also, we study the effective action itself, without any counterterms, because we are interested in the properties of "pure" effective action before applying the renormalization procedure. Of course, the main proposition $B\to f(\,\cdot-z,z)$ can be expanded on the counterterms, because, according to the general theory, they also have the covariance.

\paragraph{Acknowledgements.}
This research is supported by the Ministry of Science and Higher Education of the Russian Federation, agreement  075-15-2022-289, and by the "BASIS" foundation grant "Young Russian Mathematics".

\end{document}